\documentclass[11pt]{amsart}

\usepackage{amsmath,amssymb,amsthm,epsfig,color}
\usepackage{hyperref}
\usepackage{bbm}
\usepackage{verbatim}
\usepackage{cleveref}

\usepackage{pdfpages}
\usepackage[margin=0.8in]{geometry}

\newtheorem{theorem}{Theorem}

\title{Dual network structure of the AV node}
\date{\today}

\author{Anna V. Maltsev${}^{1, *}$}
\author{Yasir Z. Barlas${}^1$}
\author{Adina Hazan${}^2$}
\author{Rui Zhang${}^2$}
\author{Michela Ottolia${}^3$}
\author{Joshua I. Goldhaber${}^{2,3}$}

\thanks{${}^1$Queen Mary University of London}
\thanks{${}^2$Cedars Sinai}
\thanks{${}^3$University of California Los Angeles}
\thanks{${}^*$ a.maltsev@qmul.ac.uk}

\begin{document}

\maketitle

\begin{abstract}
Biological systems, particularly the brain, are frequently analyzed as networks, conveying mechanistic insights into their function and pathophysiology. This is the first study of a functional network of cardiac tissue. We use calcium imaging to obtain two functional networks in a subsidiary but essential pacemaker of the heart, the atrioventricular node (AVN).  The AVN is a small cellular structure with dual functions: a) to delay the pacemaker signal passing from the sinoatrial node (SAN) to the ventricles, and b) to serve as a back-up pacemaker should the primary SAN pacemaker fail. Failure of the AVN can lead to syncope and death. We found that the shortest path lengths and clustering coefficients of the AVN are remarkably similar to those of the brain. The network is ``small-world," thus optimized for energy use vs transmission efficiency. We further study the network properties of AVN tissue with knock-out of the sodium-calcium exchange transporter. In this case, the average shortest path-lengths remained nearly unchanged showing network resilience, while the clustering coefficient was somewhat reduced, similar to schizophrenia in brain networks. When we removed the global action potential using principal component analysis (PCA) in wild-type model, the network lost its ``small-world" characteristics with less information-passing efficiency due to longer shortest path lengths but more robust signal propagation resulting from higher clustering. These two wild-type networks (with and without global action potential) may correspond to fast and slow conduction pathways. Lastly, a one-parameter non-linear preferential attachment model is a good fit to all three AVN networks.
\end{abstract}

\section{Introduction} The atrioventricular node (AVN)  is a small structure in the heart that lies in the Koch triangle. Its function is to delay the sinoatrial node (SAN) pacemaker signal passing from the atrium to the ventricles, ensuring that the atria have enough time to eject blood into the ventricles before they contract. An additional function of the AVN is to serve as an emergency ``back-up" pacemaker in the event the SAN fails. Failure of the AVN to conduct or pace can lead to profound bradyarrhythmia or cardiac arrest.

It is growing increasingly apparent from studies in various types of tissue, including neuronal, pancreatic, and cardiac SAN \cite{stovzer2013functional, bychkov2020synchronized, liao2017small}, that healthy function emerges from  interactions of constituent parts. One important aspect of interactions is their network structure, i.e. who interacts with whom. In the brain, such a network is known as the connectome and connectomics has become a pivotal research domain. Here we present, to our knowledge, the first high-speed 2D high-resolution calcium imaging study of the AVN, adapting experimental methodology used to study the SAN by Torrente et al (2015) \cite{torrente2015burst}. Using these images, we constructed a functional connectome of the AVN in order to elucidate its network structure.

Network theory allows us to compare diverse real-world systems as networks. Many networks such as neuronal and social networks have common features: they have tight-knit communities and long-range connections. Networks with these two characteristics are called \textbf{small-world.} To be precise, the \textbf{clustering coefficient} $C$ is large (two neighbours of a node are likely to be connected) and the \textbf{average least path lengths} $L$ (the smallest possible number of nodes one has to visit on a journey from node $i$ to node $j$ averaged over all pairs $i$ and $j$) are generally short. Because $C$ and $L$ depend strongly on network size and degree distribution, one studies a normalized ratio $\gamma = C_{net}/C_{rand}$ and $\lambda = L_{net}/L_{rand}$, where $C_{net}$ and $L_{net}$ are $C$ and $L$ from the network of interest and $C_{rand}$ and $L_{rand}$ are same quantities from a random graph with the same number of vertices and degree distribution (similar to \cite{achard2006resilient, stovzer2013functional}).

Studies of network structure have been performed on several types of tissue in healthy function and disease. The literature on network properties in the brain is vast \cite{lynn2019physics, liao2017small}. This direction of study dates back to Watts-Strogatz \cite{watts1998collective}, where the concept of small-worldness was first formalized in relation to networks and universality was conjectured. The \textit{C-elegans} connectome, alongside a social network of actors and the power grid, were an inspiration for the concept \cite{watts1998collective}. In Achard et al (2006)  \cite{achard2006resilient}, the authors obtain a functional human brain network by studying correlations of frequencies using wavelet analysis, and they compute $\lambda$ and $\gamma$.
Similar studies of Islets of Langerhans in the pancreas showed that the functional connectivity network of $\beta$-cells derived from correlations of local calcium exhibits small-world properties \cite{stovzer2013functional}. They show that $\gamma$ and network efficiency (related to $\lambda$) vary depending on experimental conditions, demonstrating that the network structure adapts to function.

The AVN is known to exhibit morphological and electrophysiological heterogeneity that is conjectured to account for its two functions: slowing of conduction to the ventricles from the atria, and backup pacemaker activity \cite{aanhaanen2010developmental, Marger2011AVNEP}. Here we demonstrate heterogeneity of calcium signals, and we find that the network structure of the AVN is significantly different depending upon whether we analyze it in the presence or absence of global action potentials (AP), which we can remove using PCA (similar to PCA use in SAN in \cite{norris2023meaningful}). In the absence of global APs, the network has longer shortest path lengths (the signal propagates slower) and higher clustering (indicating local signals and modular structure). The opposite characterizes the network properties when the global AP is included in the analysis. To assess the physiological significance of this, we performed network analysis on the AVN of an atrial-specific sodium-calcium exchanger (NCX) knock-out (KO) mouse, which is known to have high degree conduction block and few spontaneous pacemaker APs compared to wild-type (WT) mice \cite{HAZAN2020329a}. We find that while its $\lambda$ is similar to the AP-present network demonstrating in practice the resilience of the small-world networks, its $\gamma$ is significantly reduced making its less small-world (see \Cref{fig:network-deg-histogram} for examples of the three types of networks extracted from the AVN).

Lastly we find that the one-parameter non-linear preferential attachment (NLPA) model  is a good fit for the AVN networks, implying that the energetic cost of longer connections in the AVN is negligible and that hubs are potentially important.
We find by brute-force parameter search that it is nearly as good as the two-parameter geometric preferential attachment (GPA) model, commonly used for the brain \cite{vertes2012simple}. We also offer a mathematical proof that the degree distribution mean and degree distribution variance are equal for a given node in any geometric graph (whose connection function has at last one moment), offering a way of quickly rejecting an infinite number of possible models based on the discrepancy between the mean and the variance which is common in real-world networks.

The importance of these findings is two-fold.
This is a first network analysis of experimentally observed cardiac tissue, and demonstration of its small-world scale-free characteristics.
With respect to neural tissue, network structure has been a hot topic of research now for 25 years, with the first paper by Watts-Strogatz carrying over 50,000 citations \cite{watts1998collective}. A vast array of methodology has been developed to further our understanding of how such network structures may arise, what role they play, and how to extract meaningful quantitative information from huge quantities of experimental data. The present paper opens the possibility of borrowing this methodology in the study of cardiac tissue. It is widely believed that a small-world architecture is so ubiquitous because it allows for an optimisation of energetic vs transmission efficiency as well as resilience with respect to vertex deletion. All these considerations will be applicable and important also for the AVN.

Secondly, the universal small-world network structures can break down in pathology and aging (\cite{liao2017small} Sections 5 and 3.3.) In particular, in schizophrenia, $C$ for the functional brain network decreases (Fig 5 of \cite{liao2017small}), which is similar to the effect of NCX KO that we find here.
The function of the AVN is to control the speed of a signal. Changes to the network structure strongly impact the speed and robustness of signal propagation through the network, thus network disruptions may have a strong impact on AVN function.

\section{Results}
Our first result is a demonstration of heterogeneity in amplitude and frequency of calcium signaling in the AVN (\Cref{fig:heterogeneity}), similar to SAN \cite{bychkov2020synchronized}. Furthermore, several of the top principal components (PCs) are delocalized (i.e. the inverse participation ratio (IPR, see SI) is low), indicating that in addition to the global action potential, there are action potentials which are semi-global (\Cref{fig:heterogeneity}) happening out of synch. This indicates unstable pacemaker activity or conduction which could result in  arrhythmia. A rough estimate of the spatial extent of the AP's is 1/IPR nodes, in this case 49.4\% of the area for the global AP and 39.4\% for the next largest AP.

Having extracted two networks from our WT data by either removing or retaining the first PC, and one from NCX KO data where in most cases the global AP is lacking, we study their $\lambda$'s and $\gamma$'s. We note that both $\lambda$'s and $\gamma$'s are significantly (one-sided $p =  0.004$ and $ 0.037$ for $\lambda$'s, $p = 0.001$ and $0.0076$ for $\gamma$'s) lower for NCX KO and WT AP-present than for WT AP-removed (see \Cref{fig:box_whiskers}). Both their $\lambda$ values are slightly larger than those of the human functional brain network ($\lambda = 1.07, \gamma = 1.23$, values read off from Fig 2 of Achard et al (2006) \cite{achard2006resilient} for mean degree 15 scale 4), the WT AP-present network has the same $\gamma$. On the other hand, $\gamma$ is significantly smaller in NCX KO (one-sided $p=0.017$) indicating a reduction of small-world attributes. The three AVN networks exhibited no significant excess kurtosis and thus are either not scale-free or too small to exhibit heavy-tailed behavior. This finding is also similar to the brain \cite{achard2006resilient}.

When controlled for mean degree (as suggested in \cite{stovzer2013functional}), the network structure of the NCX KO network is close in its properties to the WT AP-present network. The distinction of NCX KO data from WT AP-present data is in absolute magnitude of the correlation coefficients. The average correlation coefficient for two connected points is 0.94 (sd 0.02) for WT AP-present and 0.37 (sd
0.24) for NCX KO. These stronger correlations in the WT tissues are due to the presence of a global AP. The strength of the AP can be seen from the variability explained by the first PC of the data (see \Cref{fig:PCA}).

The NLPA model with just one non-linearity parameter $\alpha$ is a good fit for all the AVN networks (see \Cref{fig:model_table} for results and SI for details). We had conjectured that a more complex 2-parameter GPA model used for the brain \cite{vertes2012simple} would be appropriate. This model has two parameters $\alpha$ (how much vertices get connected to popular vertices) and $\sigma$ (how much long links are discouraged), with 8 links added at each step of growth, to match an approximate degree of 15. However, after a parametric search in $(\alpha, \sigma)$ optimizing a $p$-value for a paired t-test for standard deviation of degree distribution, $C$, and $L$, the value of the maximization parameter for the 2-parameter model was close to that from simpler one-parameter NLPA model (same as GPA with $\sigma = 0$). For WT AVN networks, the discrepancies of the optimization parameter were only 0.6\% for AP-present, 4.2\% for AP-removed, and 12.8\% for NCX KO. A pure geometric model ($\alpha = 0$) is a poor fit because the variance of degree distribution is much larger than the mean, and we prove mathematically (see SI) that the degree distribution of a node in a geometric graph with any connection function has equal mean and variance. We offer a direct comparison of key graph metrics of average from the model with the average from data.

\section{Discussion}

The two networks we identify in the AVN may correspond to the slow and fast pathways  (SP and FP) of conduction. As the anatomic structure of mouse AV node is not as well-understood as human, we can speculate that the WT AP-present network corresponds to a network of larger cells with a larger presence of high-conductance connexins CX40 as a physiological basis for the more efficient long-range network structure. This interpretation is supported by studies of CX40 KO. The KO of CX40 connexin increases the PR interval by about 20\% \cite{TEMPLE2013297}. The median increase in $\lambda$ from WT AP-present to WT AP-removed network is 20.6\%, which implies the same increase of signal propagation time (Fig 3b of \cite{watts1998collective}).
The WT AP-present network is supported on a subset of the nodes of the larger WT AP-removed network (\Cref{fig:network-deg-histogram} bottom) and may correspond to the SP. This may be mediated by medium and low conductance connexins which are known to be present in the AV node. In particular CX40 and CX43 are known to have roughly opposite expression, while CX45 is present in all the tissues of the AV node \cite{aanhaanen2010developmental,TEMPLE2013297}.

An important function of the SP is to be a fail-safe should the FP be compromised. A network with higher $\gamma$ such as the WT AP-removed network would serve this purpose, since for a network with higher clustering a lower signal transmission probability from one node to the next is needed for the signal to get transmitted through the network (compare Fig 2 with Fig 3a of \cite{watts1998collective}). This means that less powerful intracellular APs and weaker connexins are sufficient to transmit the excitation signal.

The calcium signal is well-known to be a good proxy for excitation on an intercellular level, and it also serves important intracellular functions. In NCX KO, the tissue operates in a regime of subsarcolemmal calcium loading, which likely decouples gap junctions, increases intracellular calcium, suppresses L-type Ca current and activates small K channels \cite{torrente2015burst}. The stability of $\lambda$ under such a strong intervention as an NCX KO is most likely due to the strong non-linearity in $\lambda$ as functions of the number of long-range connections (Fig 2 of \cite{watts1998collective}). The high calcium concentration decreases the intercellular gap-junction mediated connections but even a smaller number of remaining links is sufficient to have reasonably short path lengths, attesting to the resilience of small-world architecture. Indeed, despite a KO of such a crucial protein, NCX1 KO mice survive into adulthood. On the other hand, $\gamma$ is close to 1, indicating that the network is close to a random ER graph. We can speculate that the Ca overload is interfering with normal local calcium signalling, and this is reducing local links in the network.

Back-up pacemaking is an important feature of the AVN, and our data help to demonstrate a similar pacemaking mechanism as has been discovered recently for the SAN based on emergence of rhythm from interactions of cells  \cite{bychkov2020synchronized}. Additionally, a small-world modular model for the SAN has been suggested \cite{maltsev2023novel}. As we demonstrate the small-world structure on experimental data, a similar numerical model may be appropriate for the AVN.

The NSPA model approximates all three AVN networks surprisingly well with only one parameter $\alpha$. In models with $\alpha > 2$, large hub formation is expected with a few nodes connecting to all the others  \cite[Chap. 5.8]{barabasi2013network}. Thus despite finding no excess kurtosis in the degree distribution and thus no evidence of heavy tails, we expect hubs to be important in the AVN. The non-linearity present in the AVN NSPA is very strong, compared with internet and other scale-free networks (approximately $\alpha = 1$). The AVN network structure may have some differences from brain and many other networks, as the geometric component present in brain does not appear to be as important. This may be due to the smaller physical size of our data (912 $\mu$m square) and of the AVN. At such distances even long-range connections may still cheap enough  and are thus not energetically suppressed.

\section{Methods}{The AVN preparation and calcium imaging are similar to \cite{torrente2015burst} for the SAN. We select points in the image uniformly at random and threshold according to the correlation coefficient of the corresponding time series to obtain the network (see SI for details).

}

\section{Conclusions}
This is the first dynamic high-resolution  calcium imaging study of the mouse AV node and a first application of network theory to cardiac tissue. We demonstrate that the AV node functional network structure has small-world characteristics similar to a functional network of the human brain. While the average shortest path lengths are nearly unaffected by NCX KO, clustering coefficient ratio is significantly reduced, mimicking the changes in a schizophrenic brain. By removing the first principal component using PCA, we obtain a second network with larger clustering and average path lengths. We speculate that the two networks correspond to the fast and slow pathways of conduction, mediated by morphological and electrophysiological heterogeneity of the AVN. The non-linear preferential attachment model is a good fit for all AVN networks.

\section*{Acknowledgements}{A.M. thanks Victor Maltsev for useful discussions, Alexander Maltsev for suggesting CaImAn for video stabilization, and Donald Bers for a helpful discussion at the Biophysical Meeting 2024.  A.M. and Y.B. were funded by the Royal Society grant number URF$\backslash$R$\backslash$221017. J.I.G. and M.O. were funded by the National Institutes of Health grant number R01HL147569. A.H. was funded by National Institutes of Health grant number T32HL116273. This research utilised Queen Mary's Apocrita HPC facility, supported by QMUL Research-IT (http://doi.org/10.5281/zenodo.438045). The Python code used for data analytics can be found at  https://github.com/yasirbarlas/av-node-analysis}	

\section{Supplementary Information}

\subsection{Experimental methods}

\subsubsection{Intact AV node Preparation} To isolate the AV node for calcium (Ca) imaging, hearts were removed from heparinized (300 U i.p.) and anesthetized (isoflurane) atrial-specific sodium-calcium exchange (NCX) knockout (KO) mice \cite{groenke2013complete} and their NCX floxed littermates. The right and left atria, sinoatrial (SA) node and atrioventricular (AV) node were separated from the ventricles in one block, and pinned to the bottom of an optical chamber (Fluorodish, FD35PDL-100; WPI) coated with $\approx 2$ mm of clear Sylgard (Sylgard 184 Silicone elastomer kit; Dow Corning), filled with heparinized (10 U/mL) Tyrode's heated to 36 \textdegree C. The Tyrode's solution contained (in mM): 140 NaCl, 5.4 KCl, 5 Hepes, 5.5 glucose, 1 MgCl$_2$, and 1.8 CaCl$_2$ (pH adjusted to 7.4 with NaOH). Using a dissection microscope (SZX16; Olympus), the atrial tissue was transilluminated so we could identify and isolate the AV node by microdissection using the borders of the superior and inferior vena cava, the interatrial septum and the coronary sinus as landmarks \cite{torrente2015burst}.

\subsubsection{Ca Imaging of the AV node} To record cellular Ca, we immersed the AV node preparation in Tyrode's containing the high sensitivity Ca indicator Cal-520/AM (10 $\mu$M; AAT Bioquest)\cite{tada2014highly} and Pluronic F-127 (0.13\%; Invitrogen) at 20-22 \textdegree C. Imaging was  performed in dye-free Tyrodes at 34-36 \textdegree C using the xyt mode (2D) of a Leica TCS-SP5-II (Leica Microsystems Inc.), with 488-nm
excitation, $>505$-nm emission and a $10\times$ objective (N PLAN $10\times/0.25$; Leica). Scan speeds ranged from 36 to 5 ms per frame depending on the field size. Files were stored as .lif before export for analysis.

\subsection{Principal Component Analysis and Delocalization}
Here we recall the basics of PCA. PCA is a set of statistical tools rooted in random matrix theory (RMT). By studying the eigenvalues and eigenvectors of a correlation matrix of time series, one is able to separate the eigenvalues arising from noise (those that follow RMT predictions), those that reflect long and medium range spatial correlations in the data. While most eigenvalues follow RMT predictions, it is those that deviate that carry meaningful information the signal.
The eigenvectors corresponding to a few largest eigenvalues form the basis of dimensionality reduction in PCA, where they are known as principal components (PCs). The null model for PCA is the celebrated finite-rank Spiked Covariance Ensemble \cite{baik2005phase}. This consists of i.i.d. random noise variables whose eigenvalues conform to RMT predictions, plus several deterministic ``spikes", which are vectors that are designed to model a signal. There is a phase transition in the largest eigenvalues of the correlation matrix in size of the ``spike." If the spike is large enough, the largest eigenvalue separates from all the other eigenvalues, which still conform to RMT predictions. The eigenvector corresponding to this largest eigenvalue in a real data set is an estimator of the spike (known as the ``principal component estimator" \cite{stock2002forecasting}), and thus looking at the top PCs yields orthogonal components of the signal.
We perform PCA using `PCA()' from scikit-learn \cite{scikit-learn}.

One may ask about the spatial extent of these signals. To study whetheter the signals are local or somewhat global (localization and delocalization) quantitatively, one introduces the inverse participation ratio (IPR) of a vector $\textbf{v}$ of length $K$ and norm 1, i.e. $\sum_{i=1}^K \textbf{v}(i)^2 = 1$, as
\begin{equation}\label{e:IPR}
\text{IPR}(\textbf{v}) = \sum_{i=1}^K \textbf{v}(i)^4.
\end{equation}
We notice here that the IPR can range between $1/K$ for vectors which are completely delocalized with each component equalling $\frac{1}{\sqrt K}$ and 1 for vectors which are completely localized with one component equal to 1 and the rest equal to 0. Since correlations are mostly positive, the eigenvectors are also strongly skewed toward the positive, which means that for the eigenvectors to be orthogonal they can only be large on non-intersecting sets and near zero everywhere else. Therefore a rough estimate of the number of vector components that are large is $1/\text{IPR}$, which gives $K =$ the number of non-zero elements in the vector.

\subsection{Network Construction and Threshold Analysis}
To extract a network from calcium imaging data, we first use CaImAn \cite{giovannucci2019caiman} software for motion stabilization. We select 150 small regions in the image uniformly at random, conditional that they do not intersect, and we average the signal over the region for each video frame. This yields 150 time series. We standardise this matrix by removing the mean and scaling to unit variance. This is done using the `StandardScaler()' function in the scikit-learn library \cite{scikit-learn}. Then for WT AVN we can construct two different networks: one where the 1st PC is present and 2nd one where the 1st PC is removed. In each of these two settings we construct the correlation matrix
and threshold according to the correlation coefficient. We choose the threshold so as to obtain an approximate mean degree of 15 in the largest connected component. This procedure gives us an effective connectivity network of the AV node.
For all our degree 15 networks, the correlations marked with edges are statistically significant. We assess significance similar to equation (2) in \cite{stovzer2013functional} for pancreas. For all network analysis, the Python package NetworkX \cite{hagberg2008exploring} was used.

The mean degree 15 thresholding was chosen to balance between network sizes and significance of correlations in the network (see \Cref{fig:mean-deg-dependence}). When mean degree is 20, one network in one of the datasets has a $p$-value of 0.1 which is too high. When mean degree is 10, the sizes of the largest connected component, which is what we study, get small.
That the results of our study are the same across mean degrees 10, 15, and 20. The NCX KO $\gamma$ is still significantly smaller than that of WT AP-present network, and similarly the $\lambda$ and $\gamma$ are significantly higher for the WT AP-removed network than for the other two (see \Cref{fig:box-whisker}).

\subsection{Details of Our Model and the Inappropriateness of Competing Models}

\subsubsection{Our model} We study the applicability of the geometric preferential attachment model \cite{yook2002modeling, vertes2012simple} to the AVN networks. Since our networks are of different sizes, optimization parameter will be different from \cite{vertes2012simple}, though we take inspiration from them in their use of aggregating t-test statistics. For each of our seven data points for WT and nine data points in KO, we read the number of nodes in the network. Then we generate a geometric preferential attachment model for each number of nodes, and we save the standard deviation of its degree distribution, $C$, and $L$. Then construct a paired t-test statistics for the difference between the simulated graph and the data in the respective columns, and read off three different p-values corresponding to standard deviation of degree distribution, clustering coefficient, and average shortest path lengths. We multiply them together to form an aggregate optimization parameter, and we for each value of $\alpha$ and $\sigma$ we do this 1000 times, and average the optimization parameter the 1000 runs. The brute force optimization was run on Queen Mary University High Performance Computing Apocrita Cluster. The result of these simulations are the values of $\alpha$ and $\sigma$ that maximize the optimization parameter. The optimization was run for $0 \leq \alpha \leq 12$ and $0 \leq \sigma \leq 10$ with increment 0.1 for all the networks.
However, for WT AVN AP-present and NCX KO networks, the largest value maximization parameter with $\sigma = 0$ (a pure non-linear preferential attachment model) is only  0.6\% for AP-present, 4.2\% for AP-removed and 12.8\% for NCX KO away from the $p$-value for the $(\alpha, \sigma)$-maximizer respectively, so we opt for the simpler non-linear preferential attachment model. The optimal $(\alpha, \sigma)$ values were (8.7, 0.7) for AP-present, (5, 0.8) for AP-removed, and (5.5, 2.1) for NCX KO.

\subsubsection{Other Candidate Models}

The Erd\"os-Renyi model has small clustering coefficients (equal to the network density), while the AVN networks have much larger clustering, thus it is not an appropriate model for the AVN.
The variance of the degree distribution in AVN networks is very large. It is larger than that of Watts-Strogatz model and Erd\"os-Renyi model. A soft random geometric graph is constructed by placing vertices uniformly at random in a unit square, and connecting points with a probability that is a function of the distance between them. Here we offer a mathematical proof that in a geometric graph with any connection function,  given a node, its expected degree is equal to the variance of its degree distribution. While this is not the same as mean degree of the graph and the variance of the degree distribution of the graph due to dependence between vertex degrees, it is suggestive that these quantities will be reasonably close. However, in the AVN networks the variance of the degree distribution is vastly larger than the mean
(see \Cref{fig:vars}).

\begin{theorem}
Let $G = \{V, E\}$ be a graph with vertex set $V$ and edge set $E$. Suppose the vertices form a Poisson process in the plane with intensity $\lambda$ and two vertices with coordinates $v, w$ are connected  independently of any other connections with probability $f(\|v-w\|)$, where $\|\dot\|$ is Euclidean distance and $f$ is differentiable and such that $f(x)$ and $x f(x)$ are integrable on the half-line. Let $D(v)$ be the degree of vertex $v$. Then
\begin{equation}\label{e:E=var}\mathbb{E} D(v)= \text{var} D(v) = 2 \pi \lambda \int_0^\infty r f(r) dr
\end{equation}
\end{theorem}
\begin{proof}
Let us decompose the plane into disjoint concentric annuli of width $\varepsilon$ centered at $v$ and number them from 1 to $\infty$, and call them $A_i$. Because the vertices are drawn from a Poisson process the number of vertices $k_i$ in each $A_i$ is an independent random variable distributed according to the Poisson distribution with parameter $\pi \lambda (r+\varepsilon)^2 - \pi r^2) = 2 \lambda \pi r \varepsilon + O(\varepsilon^2)$. Also the number of vertices inside $A_i$ that are connected to $v$, call it $k_{i, v}$, are independent random variables. Since $$D(v)= \sum_{i=1}^{\infty} k_{i,v}$$
and $k_{i,v}$ are independent, we have that
\begin{align*}
\mathbb{E}D(v) &= \sum_{i=1}^{\infty} \mathbb{E} k_{i,v} \\
\text{var}D(v) &= \sum_{i=1}^{\infty} \text{var}
(k_{i,v}).\end{align*}
A differentiable function $f(x)= f(r)+ O(\varepsilon)$ for any $x\in [r, r+\varepsilon]$. By Poisson thinning, $k_{i, v}$ has Poisson distribution with parameter $2 \pi r \varepsilon f(r) \lambda + O(\varepsilon^2)$, which has variance and mean equal to $2 \pi r \varepsilon f(r) \lambda + O(\varepsilon^2)$, thus summing over $i$ and taking $\epsilon$ to 0 we \eqref{e:E=var}.
\end{proof}

\section*{Accopanying files}
\begin{itemize}
\item WT.avi: Example of Calcium signalling video in wild-type tissue.

\item KO.avi: Example of Calcium signalling in NCX KO tissue.

\item WT\_AP\_present.xlsx: Network metrics for WT AP-present networks for mean degrees 10, 15, 20, and 25.

\item KO.xlsx: Network metrics for KO networks for mean degrees 10, 15, 20, and 25.

\item WT\_AP\_removed.xlsx: Network metrics for WT AP-removed networks for mean degrees 10, 15, 20, and 25.
\end{itemize}

\bibliographystyle{plain}
\bibliography{references}
\newpage

\begin{figure*}
\centering
\includegraphics[width=.32\textwidth]{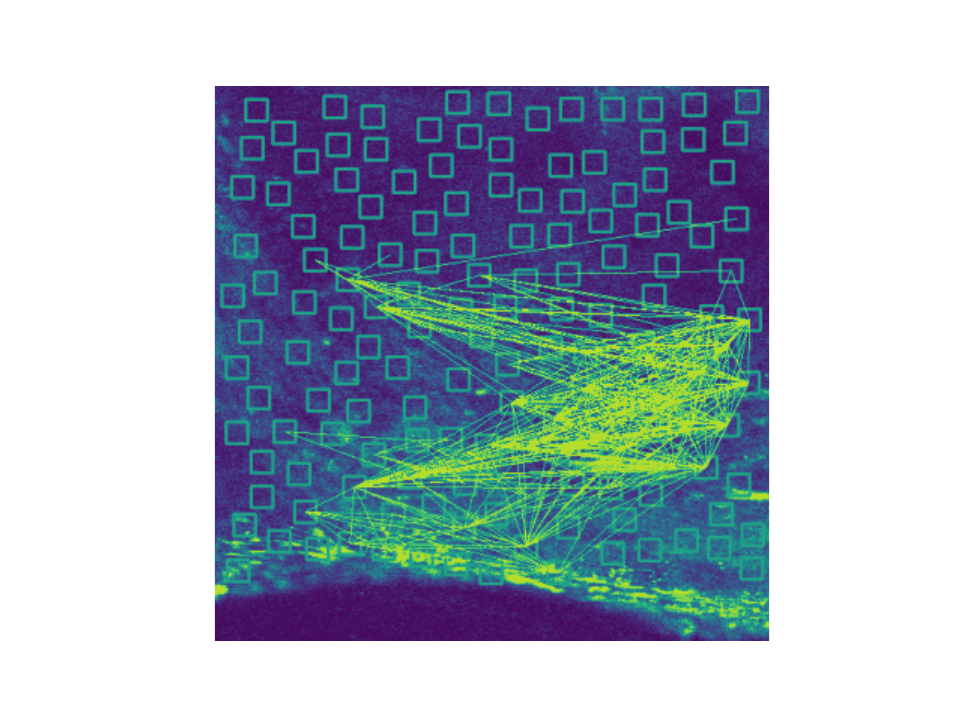}
\includegraphics[width=.32\textwidth]{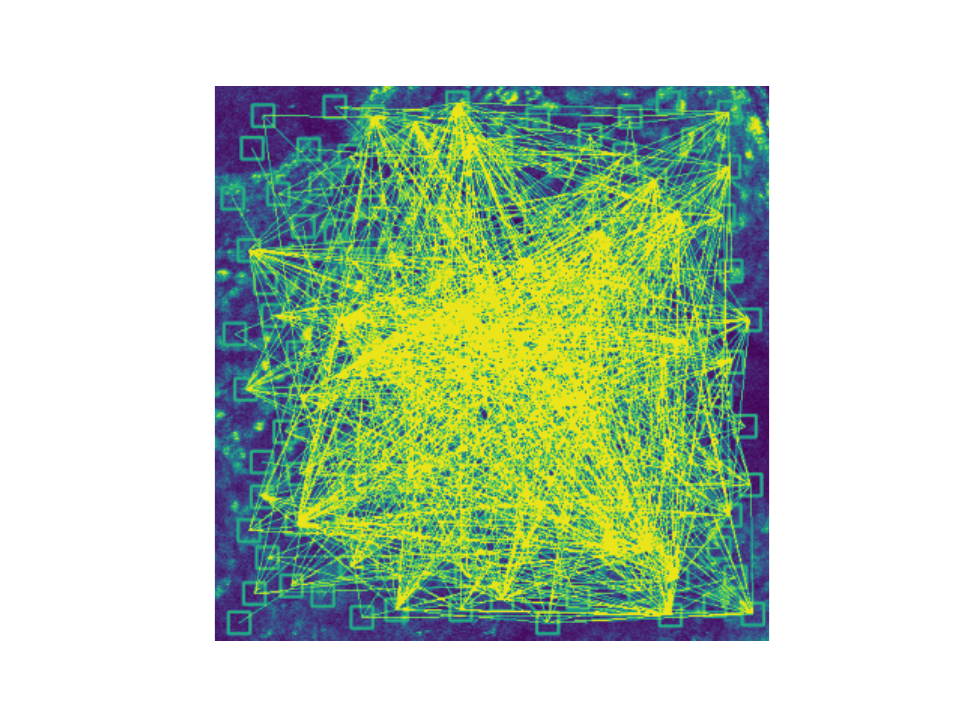}
\includegraphics[width=.32\textwidth]{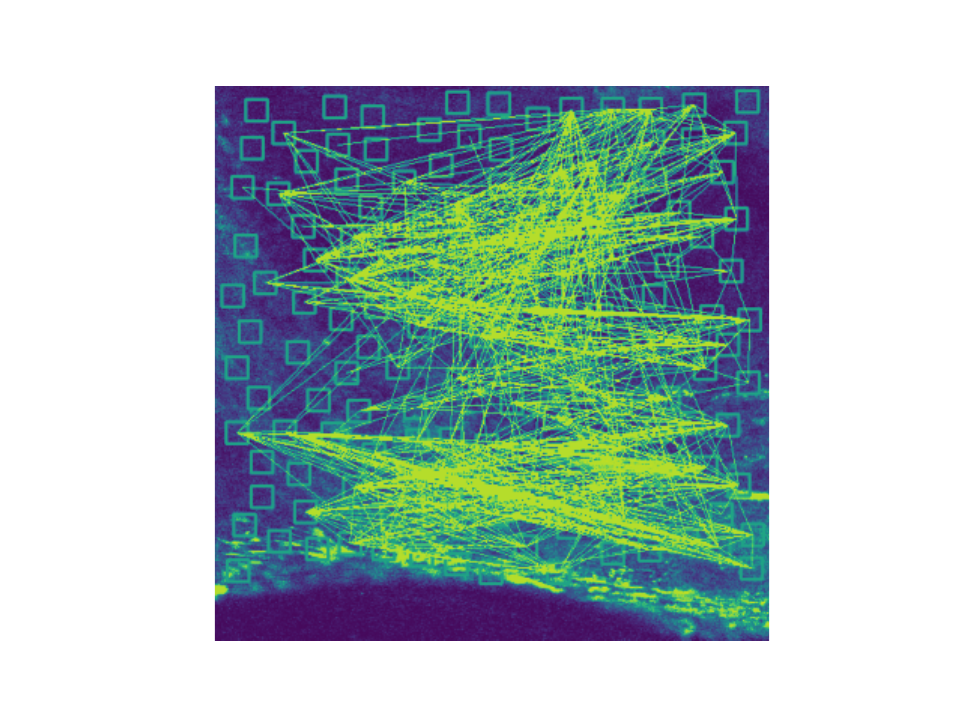}
\includegraphics[width=.32\textwidth]{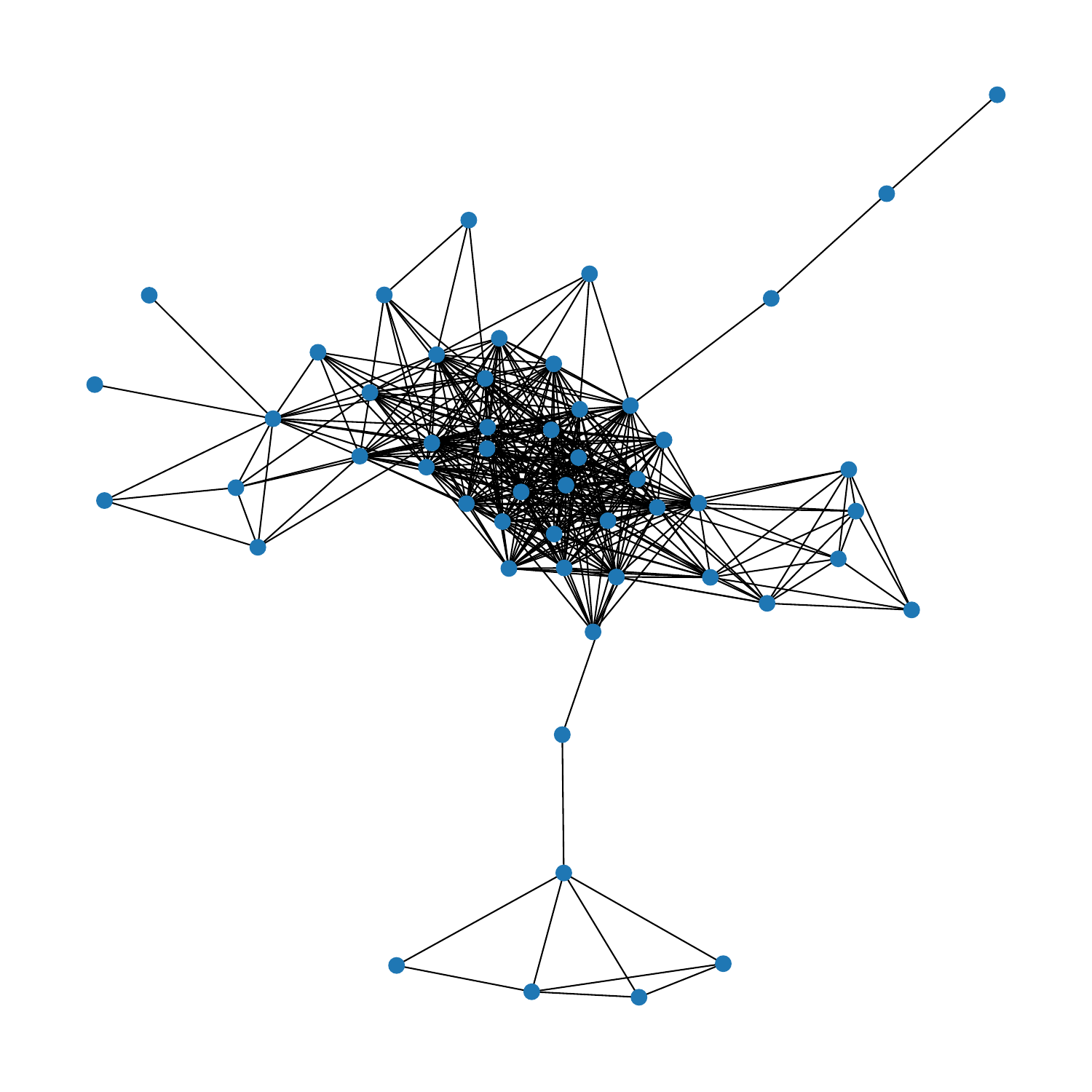}
\includegraphics[width=.32\textwidth]{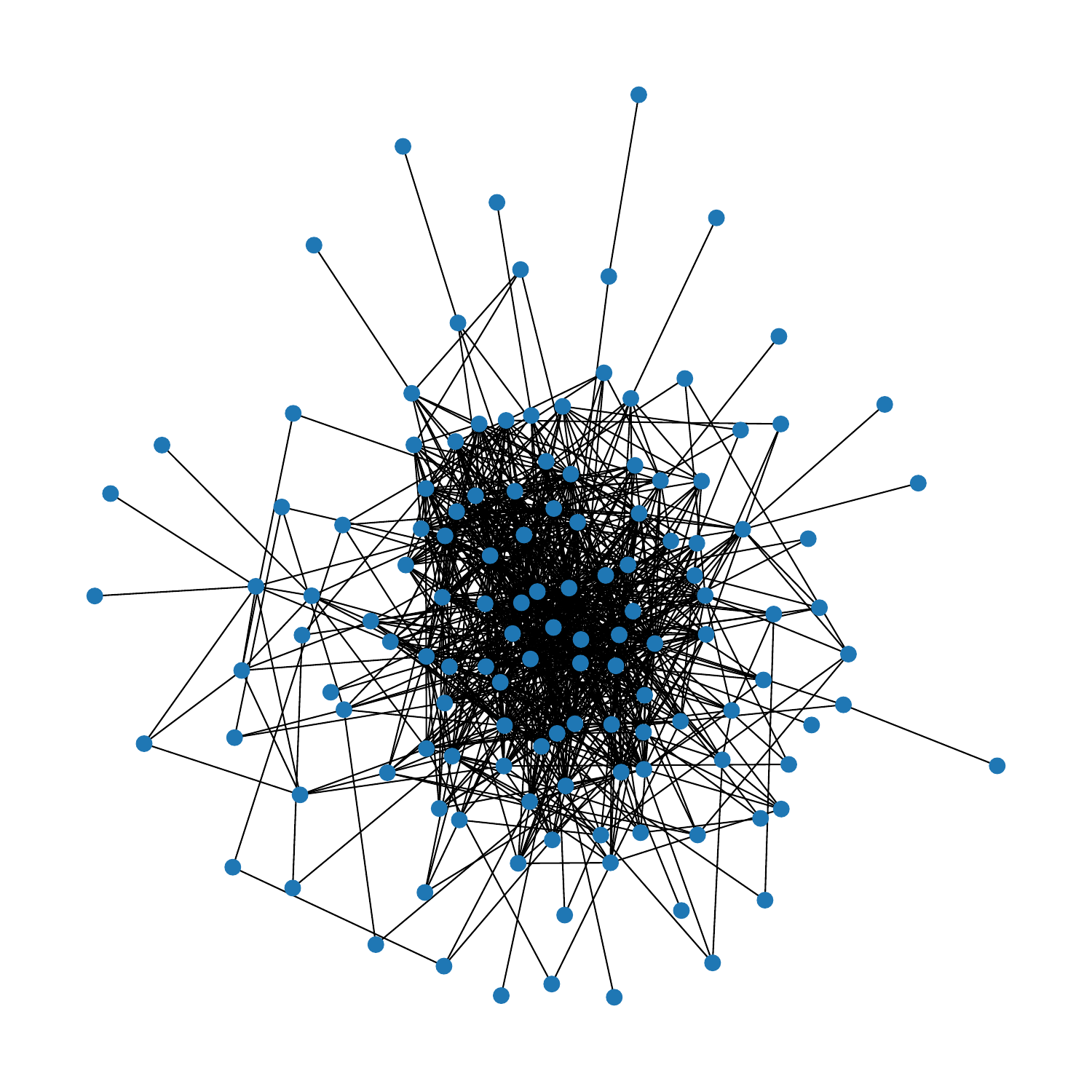}
\includegraphics[width=.32\textwidth]{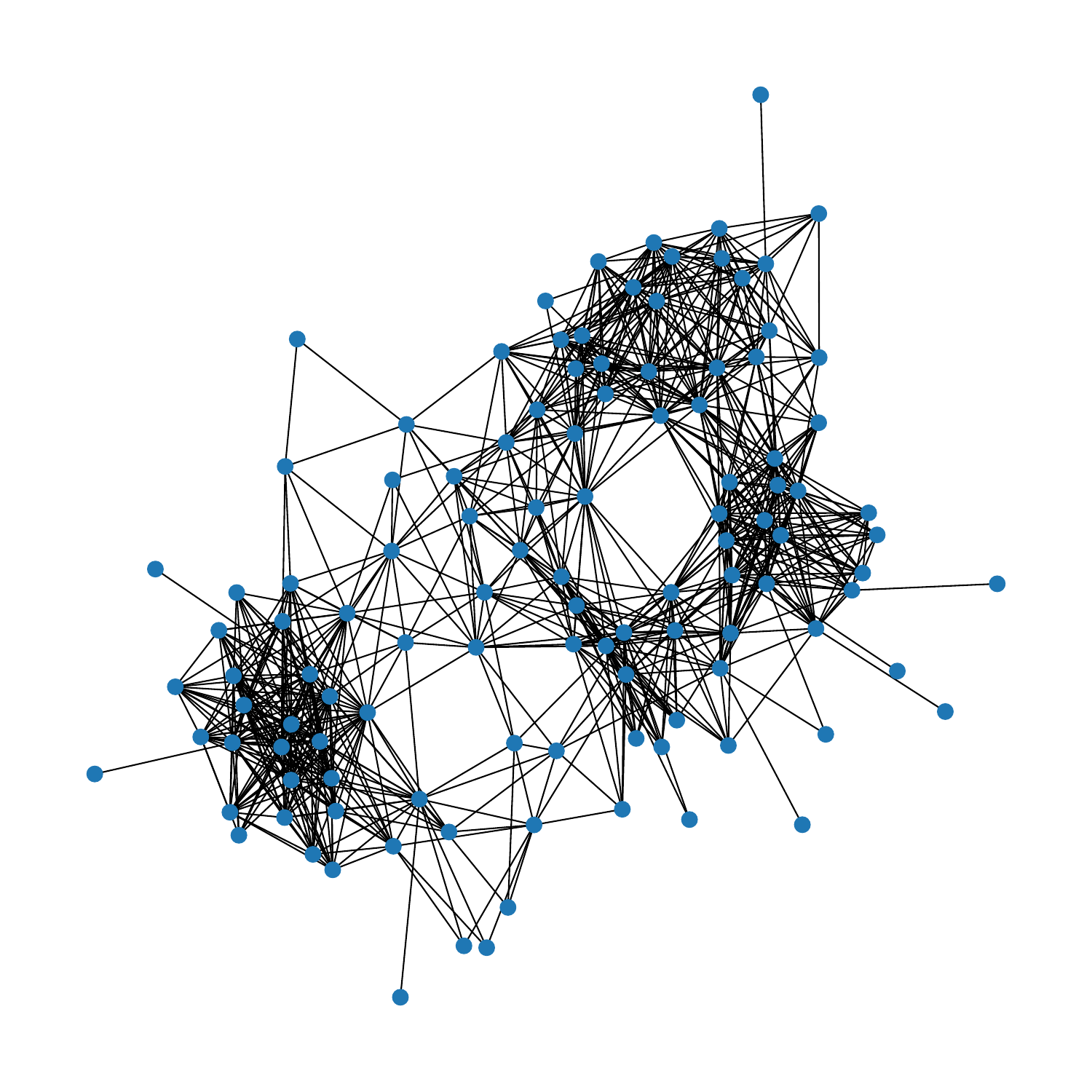}

\caption{Examples of networks presented as networks (bottom) and as embedded in the AV node video (top) for WT AP-present(left), and NCX KO (middle), and WT AP-removed (middle).}\label{fig:network-deg-histogram}
\end{figure*}

\begin{figure}[h]
\centering
 \includegraphics[width=.6\linewidth]{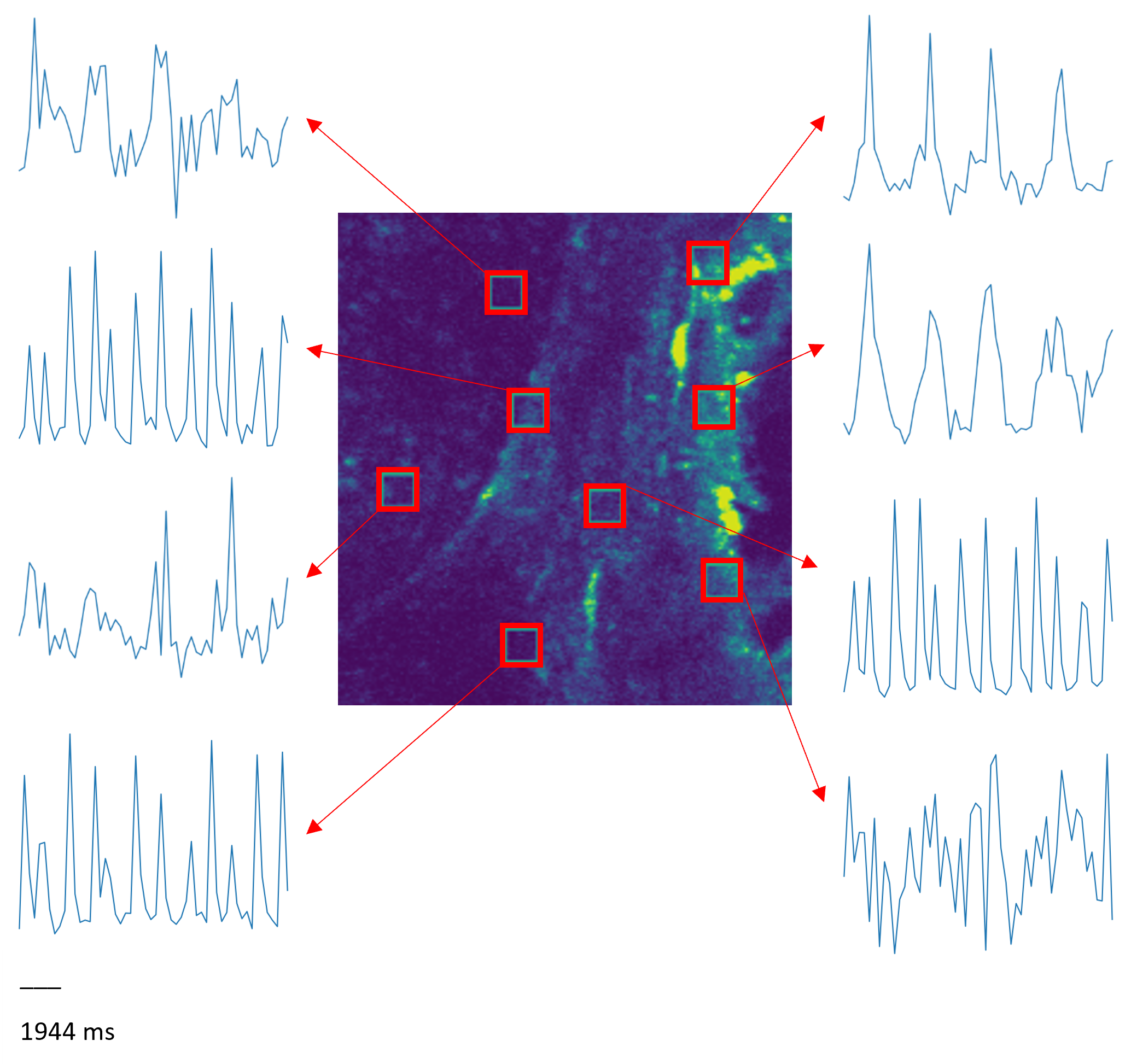}
 \includegraphics[width=.6\linewidth]{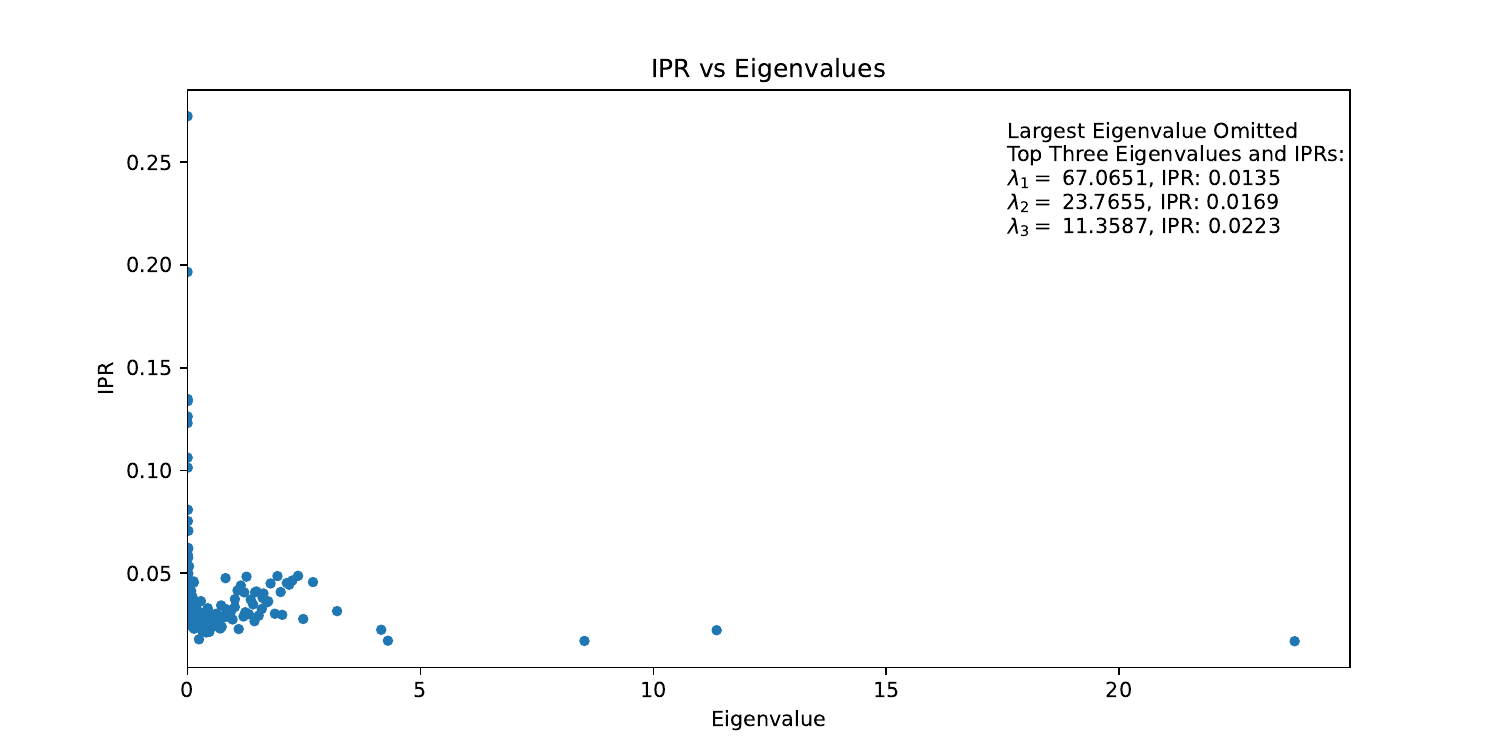}\\
\includegraphics[width=.3\linewidth]{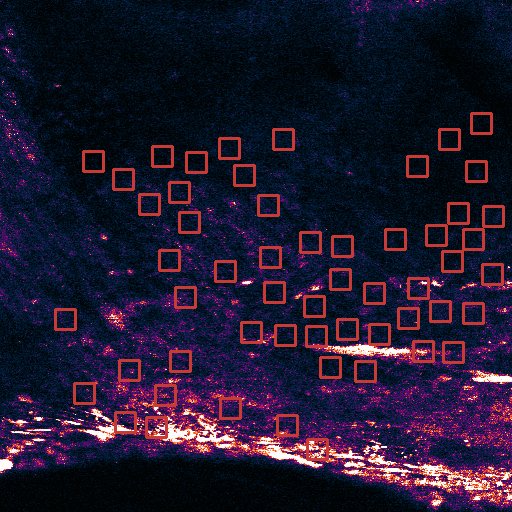}
\includegraphics[width=.3\linewidth]{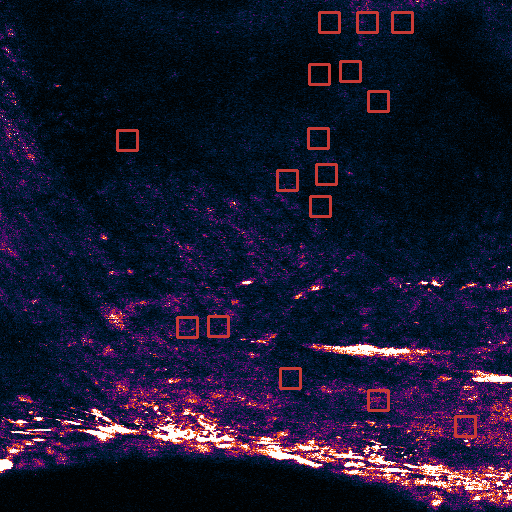}

\caption{(top) Snippets of calcium time series from different regions of the AV node demonstrating heterogeneous calcium signaling; (middle) Example of IPR vs eigenvalue scatter plot for a calcium imaging video; (bottom) Regions of the AV node where the 1st PC (left) and the 2nd PC (right) are larger than 1.7sd, see SI for details, indicating where the global AP and the next largest semi-global AP are localised}
\label{fig:heterogeneity}
\end{figure}

\begin{figure}[h]
\centering
\includegraphics[width=.7\linewidth]{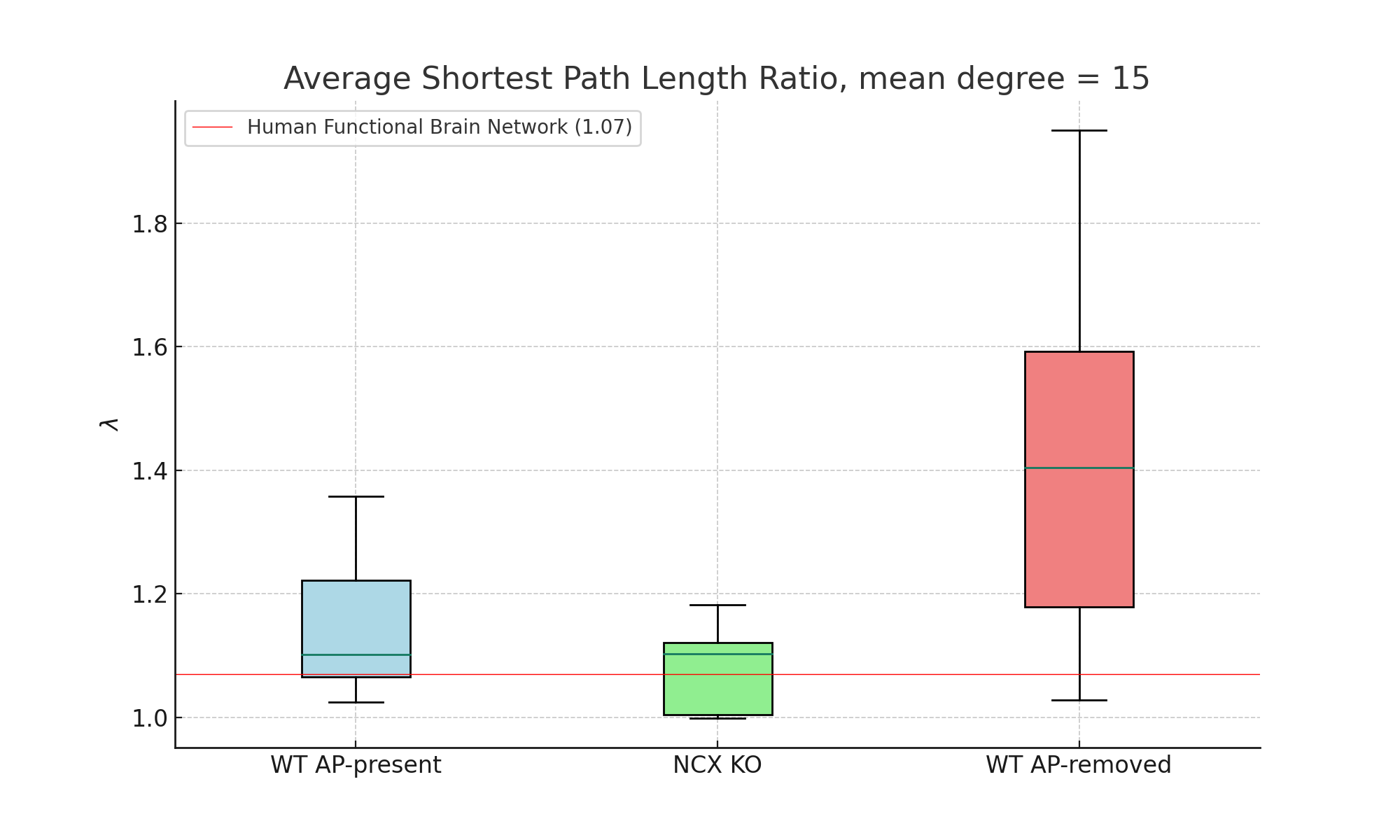}
\includegraphics[width=.7\linewidth]{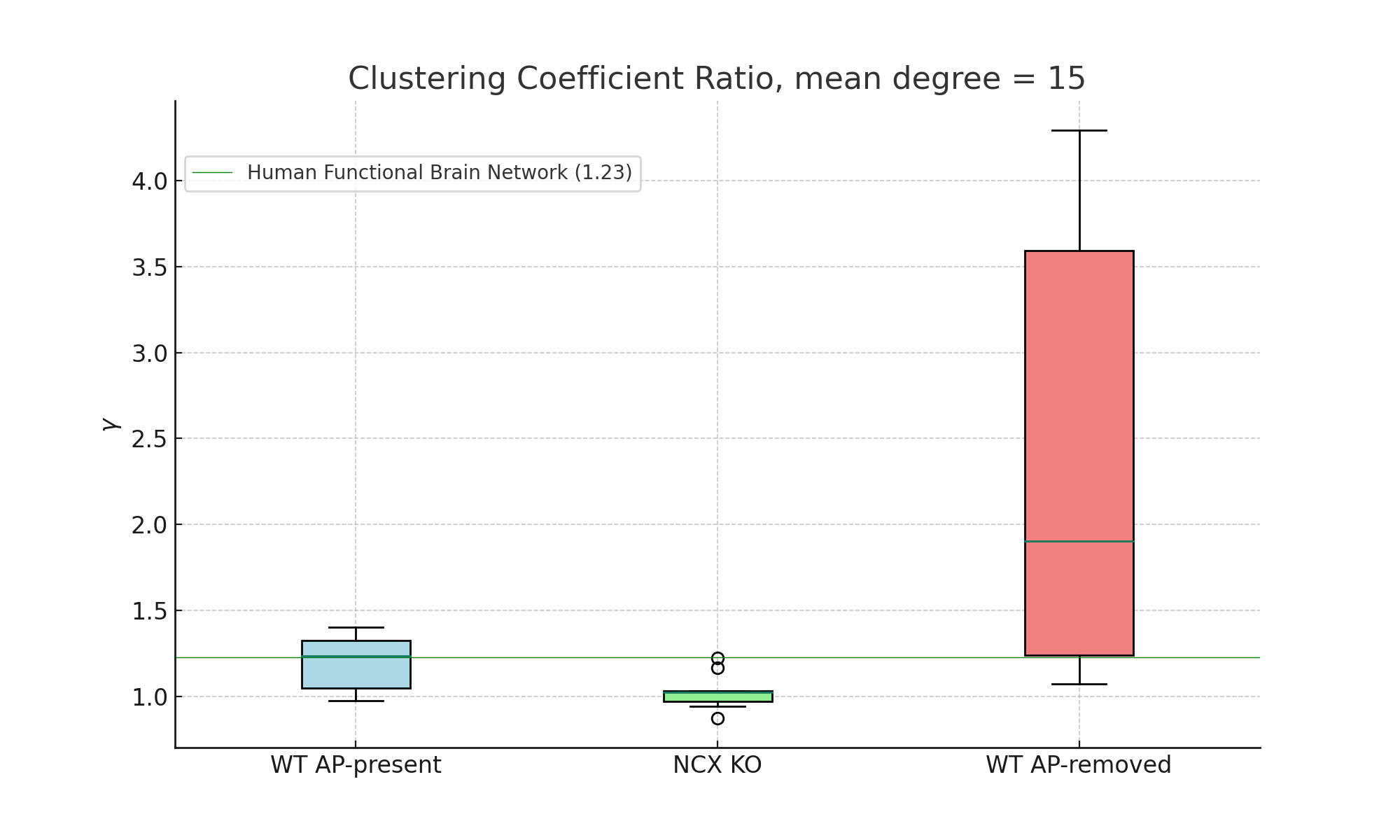}
\caption{Comparisons of $\lambda$ between NCX KO, WT AP-present, and WT AP-removed}
\label{fig:box_whiskers}
\end{figure}

\bigskip
 \begin{figure}[h]
\centering

 \begin{tabular}{c |c c c c  }
network fitted & $\alpha$ &  sd degree & $C$ & $L$\\
WT AP-present& 8.4  & 8.07 & 0.754 & 1.93\\
& & 12.8\%, 0.42 & 0.2\%, 0.03 & 3.2\%, 0.19\\
WT AP-removed & 4.9 &11.0 & 0.682 &1.99 \\
& & -20.5\%, 1.0 & 1.0\%, 0.10 & 30\%, 1.12 \\
NCX KO & 4.1  & 10.8  & 0.606  & 2.03\\

 & & 17\%, 0.86  & -8.0\%, -0.24 & 5.2\%, 0.38
\\
 \end{tabular}
 \caption{Comparison of model parameters fit to various networks. The columns ``mean degree," ``sd degree", $C$, and $L$ contain the corresponding metrics from the fitted model (1st line) and two measures of discrepancy (2nd line): percent discrepancy between the real network and the fitted model (network - model)/network and number of sd of the data. Brute-force optimization was based on optimizing $t$-statistics of the columns similar to \cite{vertes2012simple}.
\label{fig:model_table}}
\end{figure}

\begin{figure}[h]
\centering
 \includegraphics[width=.8\linewidth]{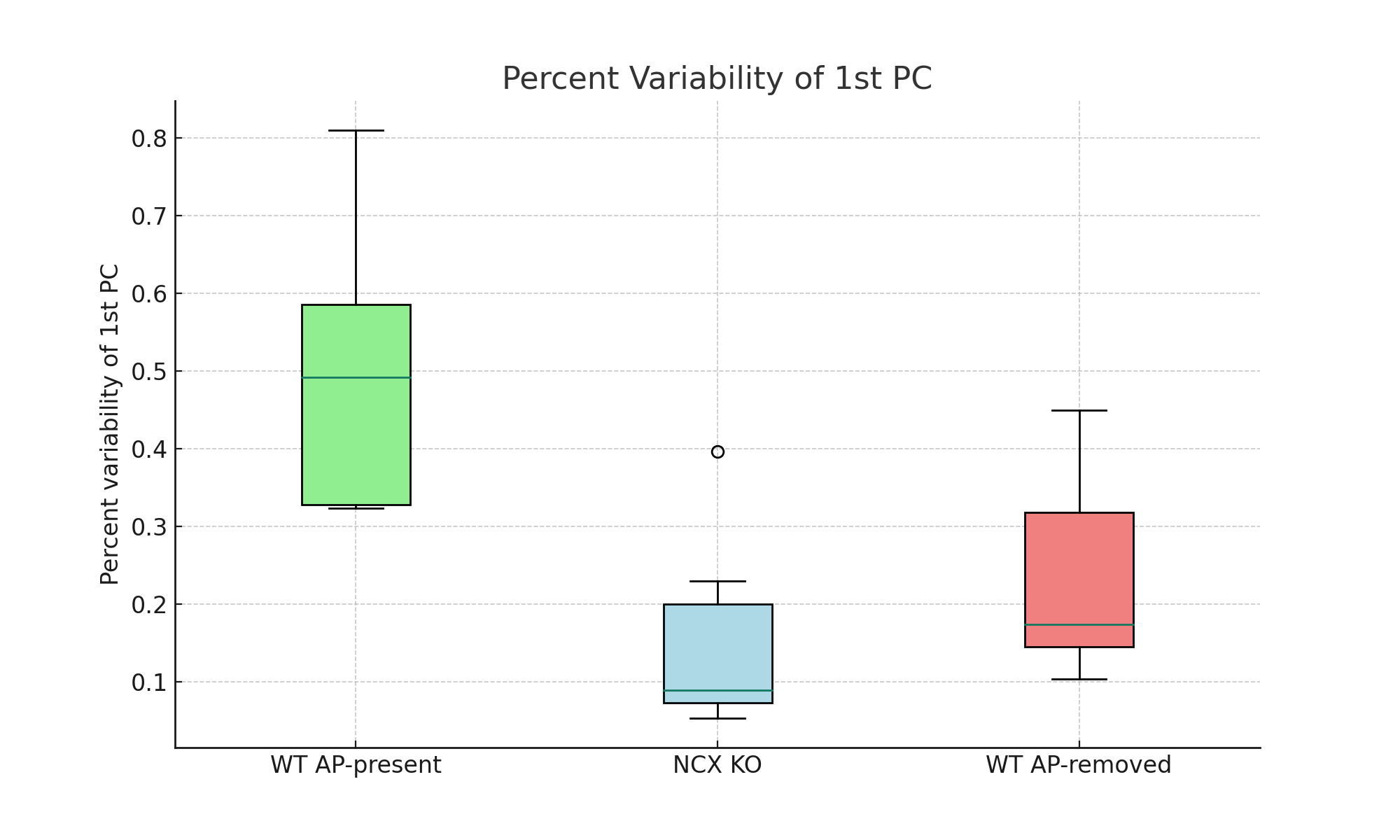}
\caption{Percent of the data variability explained by the 1st PC }
\label{fig:PCA}
\end{figure}

\begin{figure}[h]
\centering
 \includegraphics[width=.49\linewidth]{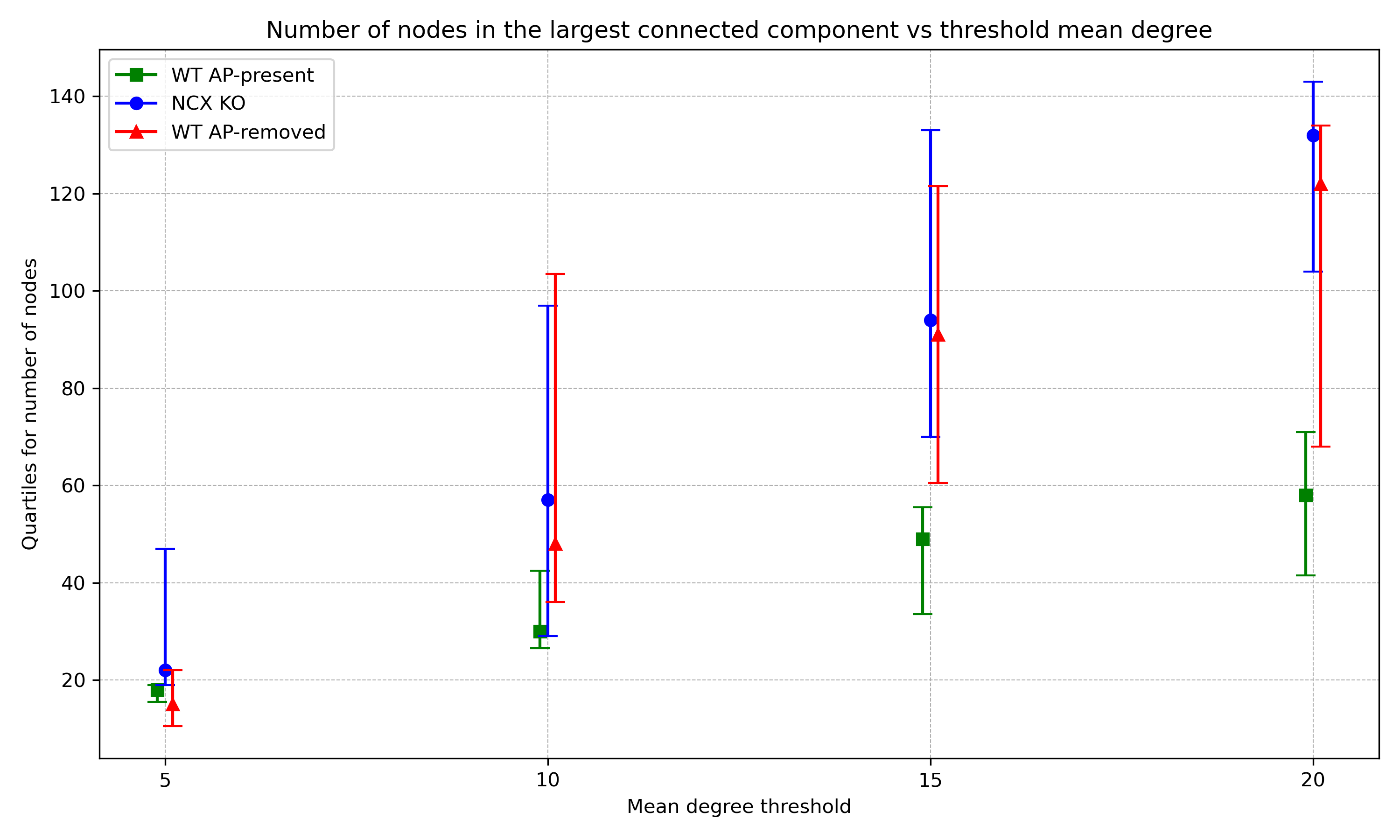}
 \includegraphics[width=.49\linewidth]{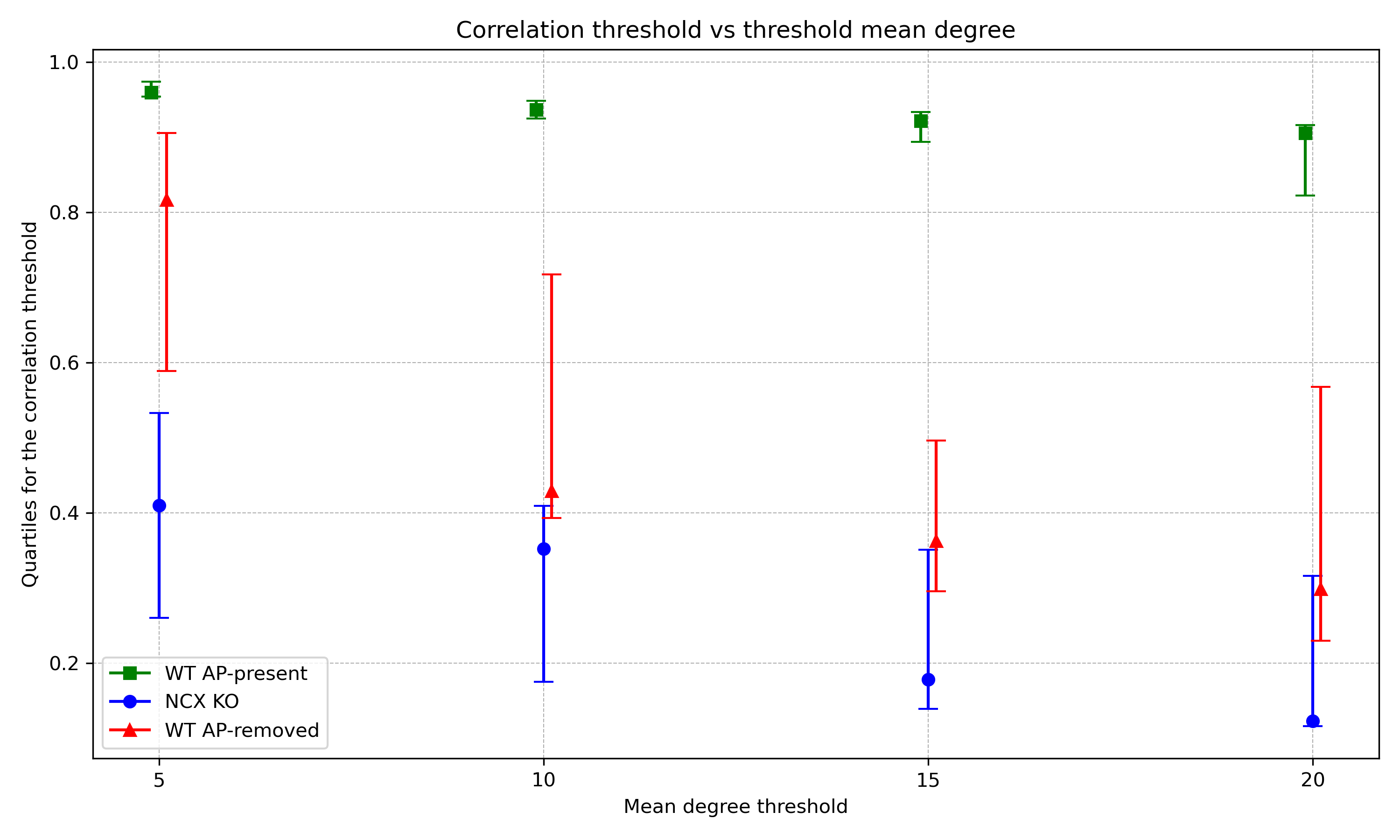}
 \includegraphics[width=.49\linewidth]{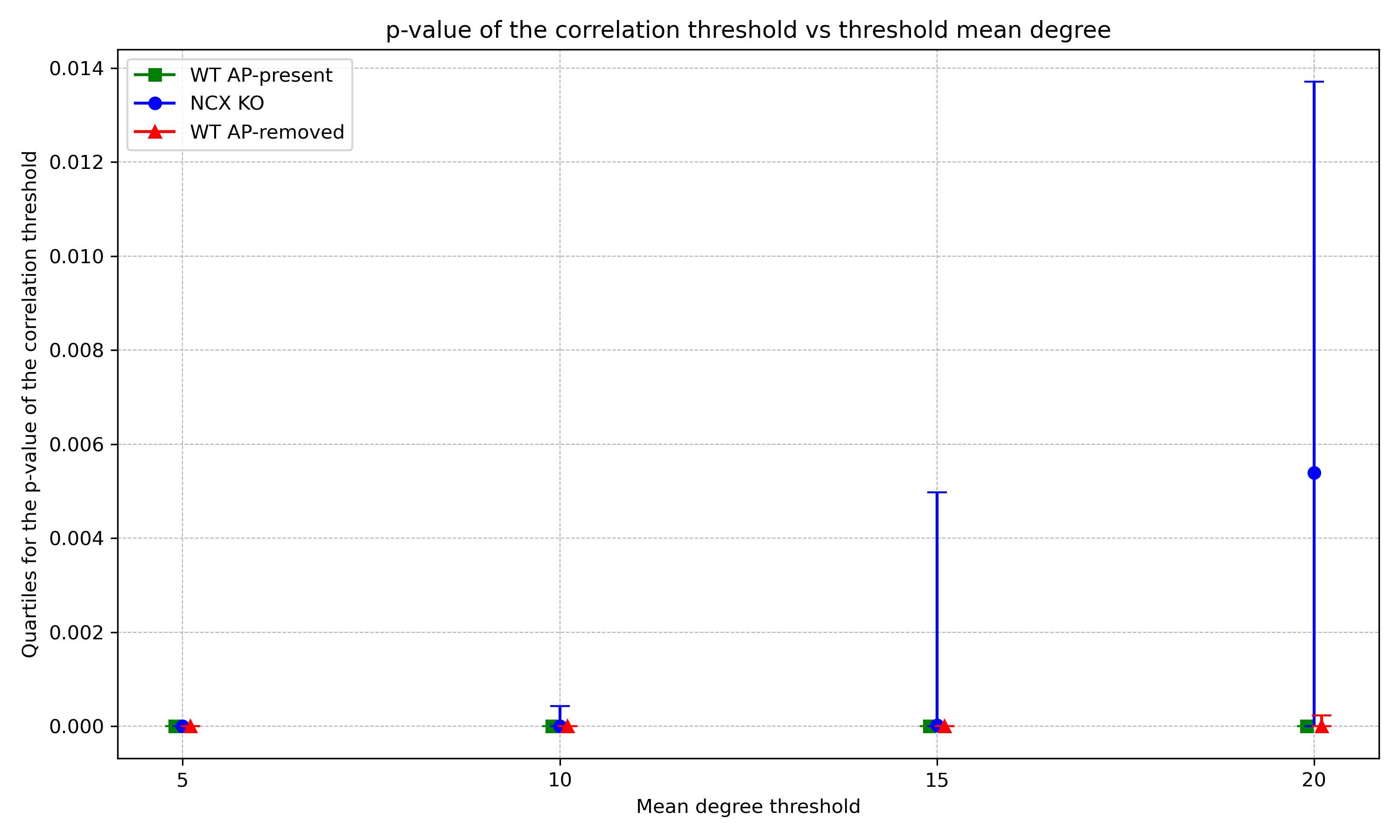}
\caption{Dependence of various network parameters on the mean degree threshold}
\label{fig:mean-deg-dependence}
\end{figure}

\newpage

\begin{figure*}
\centering
\includegraphics[width=.49\textwidth]{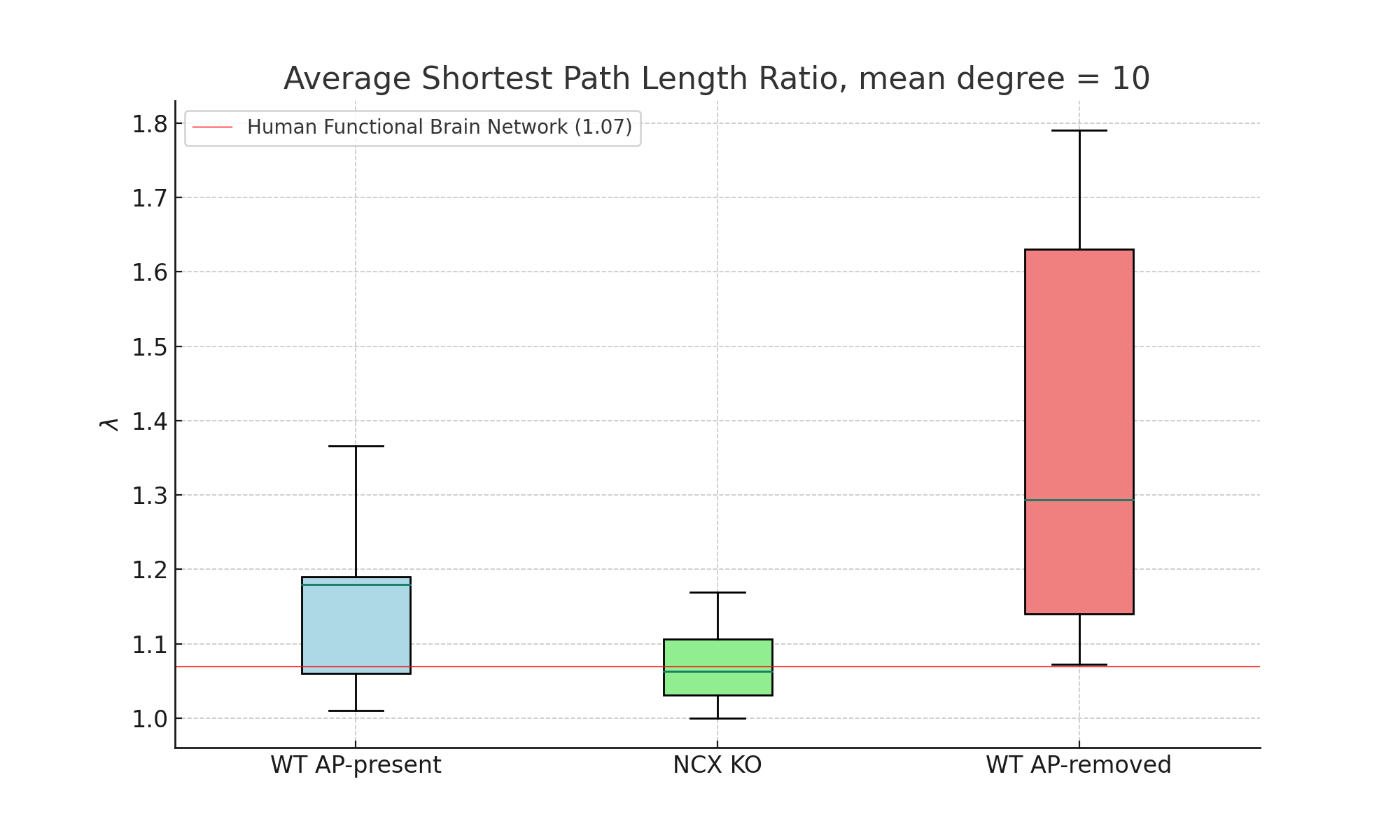}
\includegraphics[width=.49\textwidth]{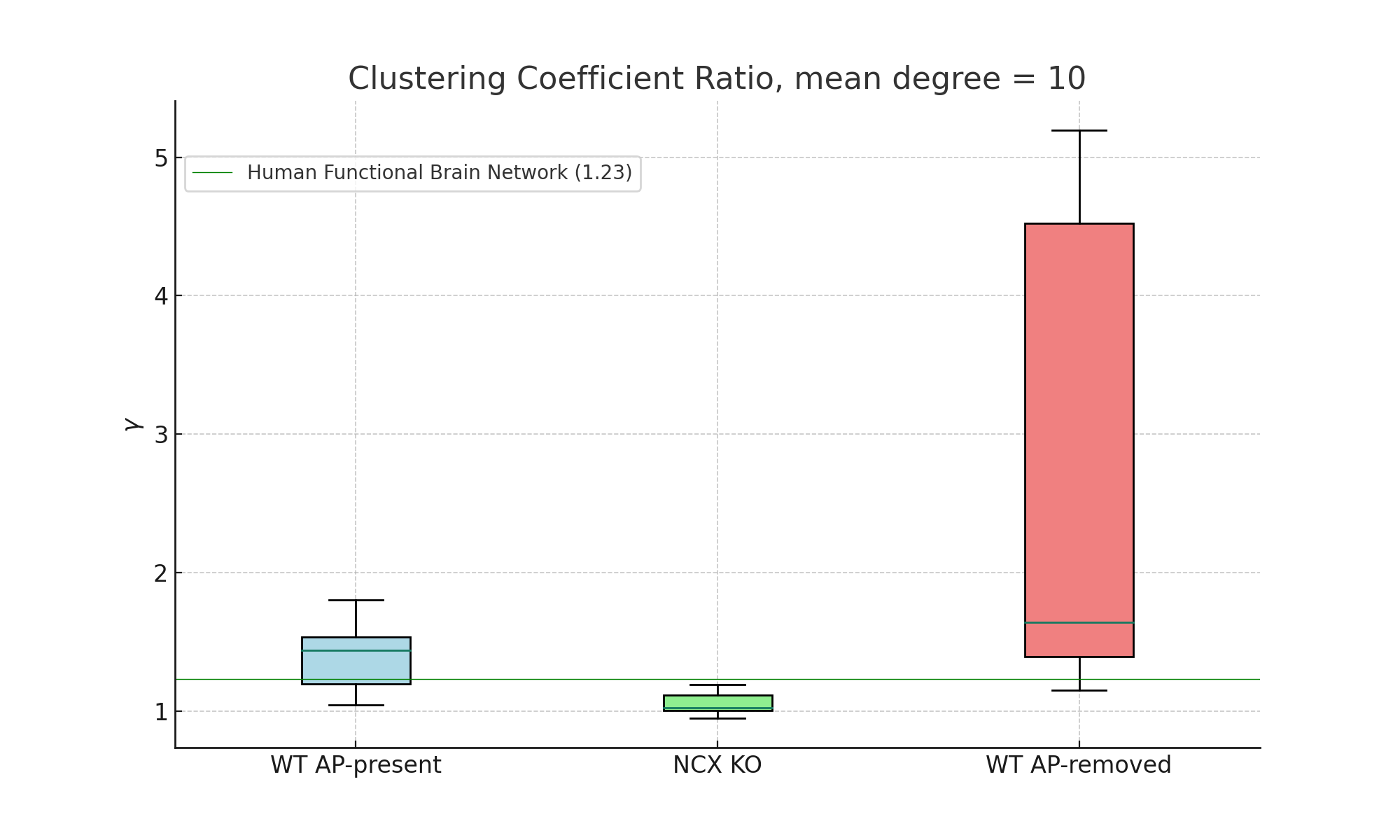}
\includegraphics[width=.49\textwidth]{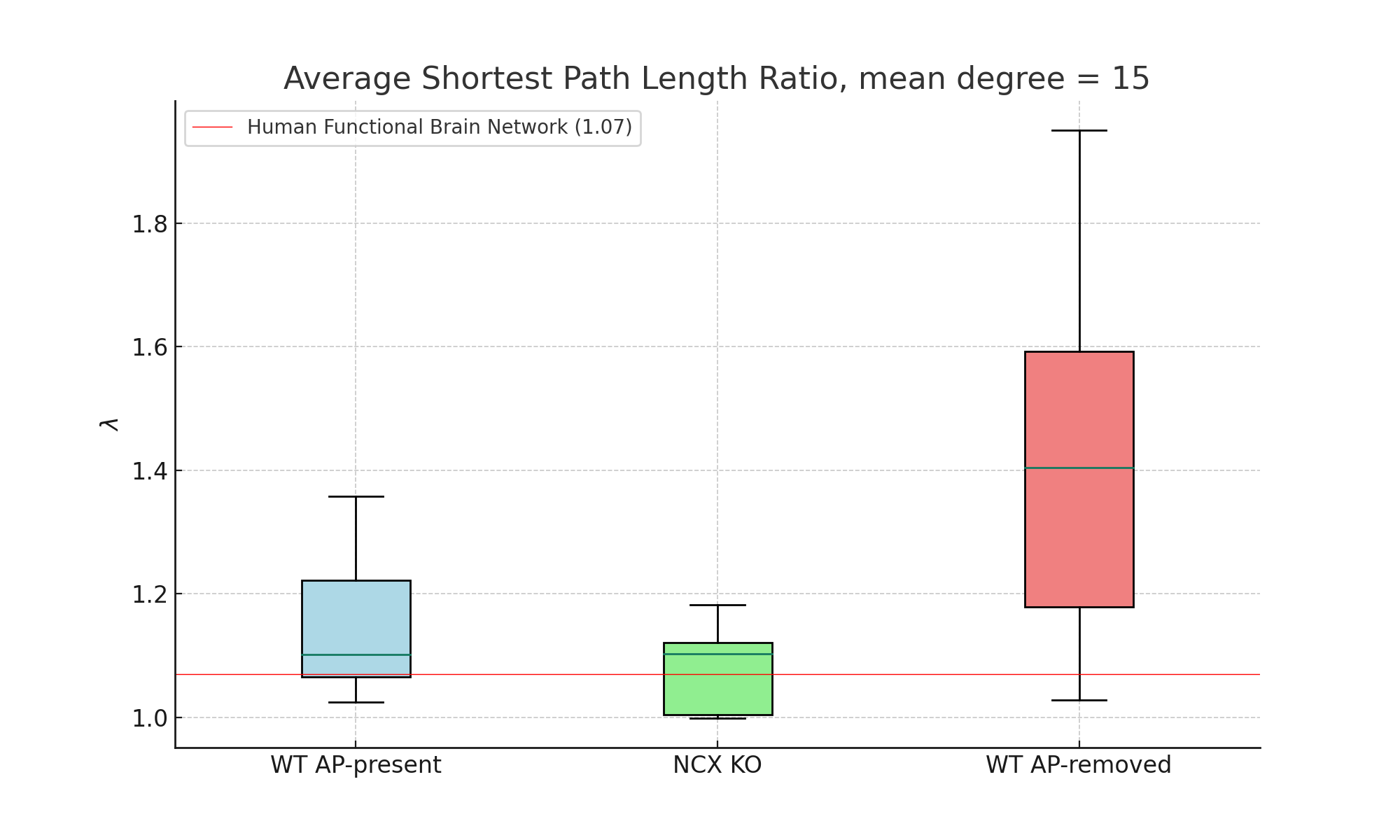}
\includegraphics[width=.49\textwidth]{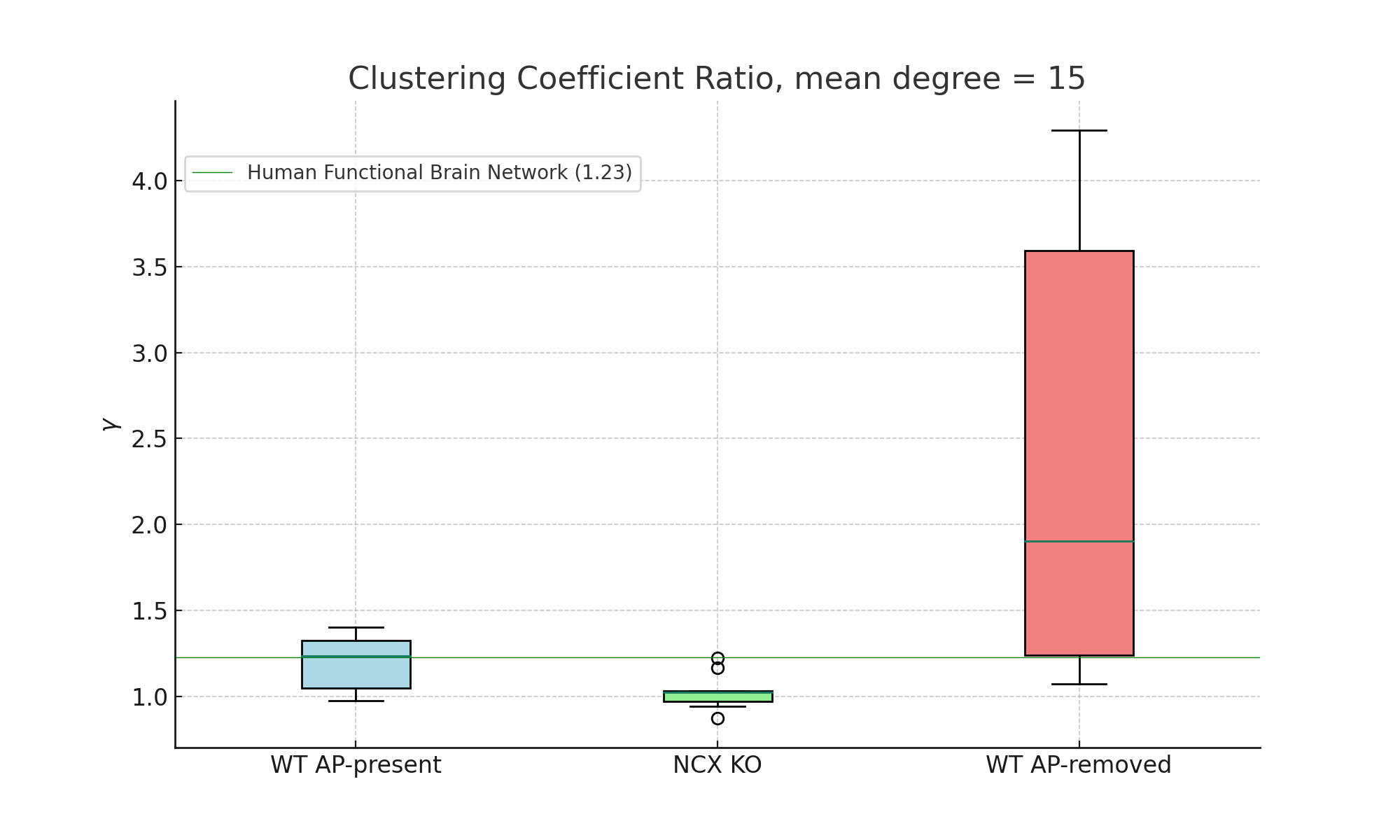}
\includegraphics[width=.49\textwidth]{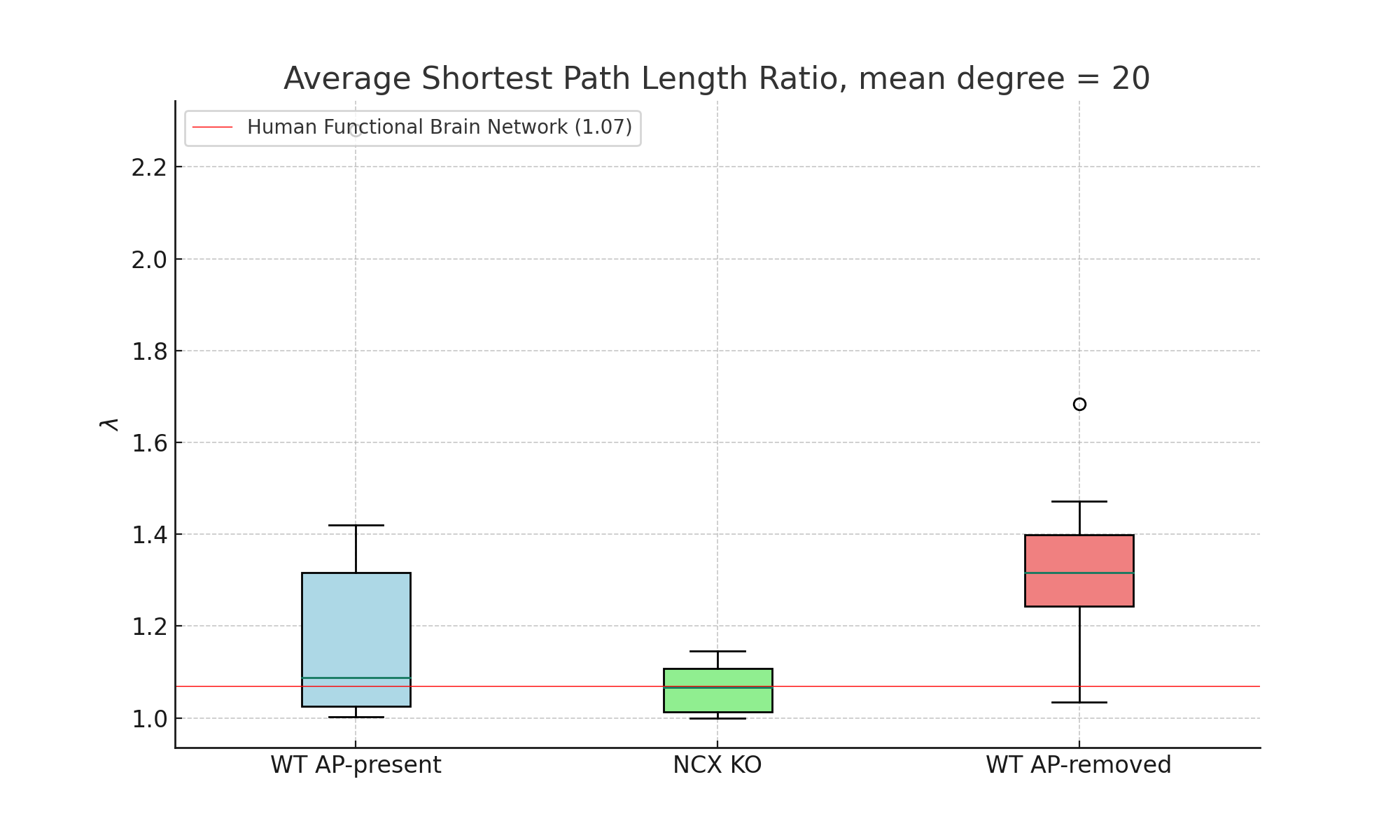}
\includegraphics[width=.49\textwidth]{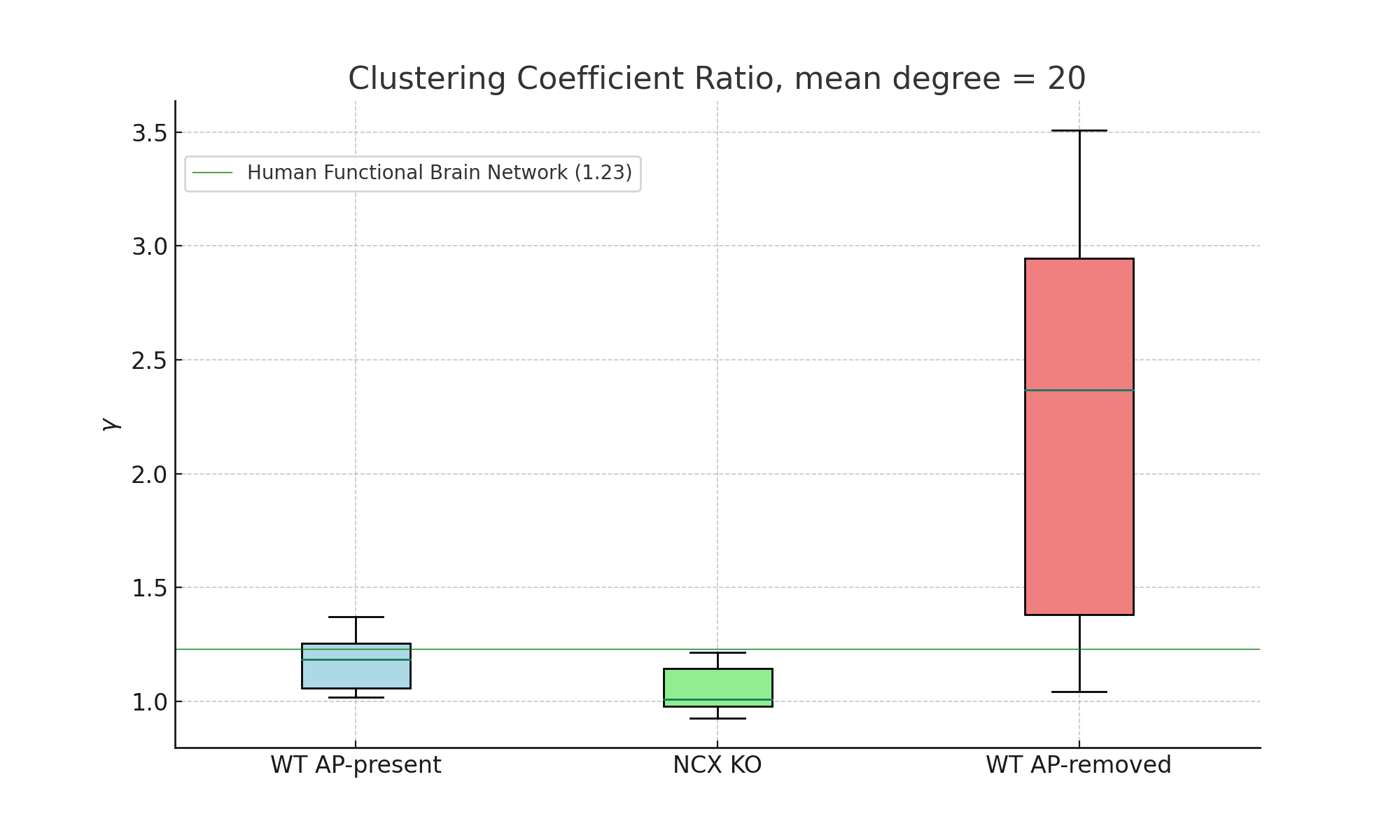}

\caption{Box-and-whisker plots for $\lambda$'s (left) and $\gamma$'s (right) for mean degrees 10, 15, and 20 thresholding}\label{fig:box-whisker}
\end{figure*}

\begin{figure}[h]
\centering
 \includegraphics[width=.8\linewidth]{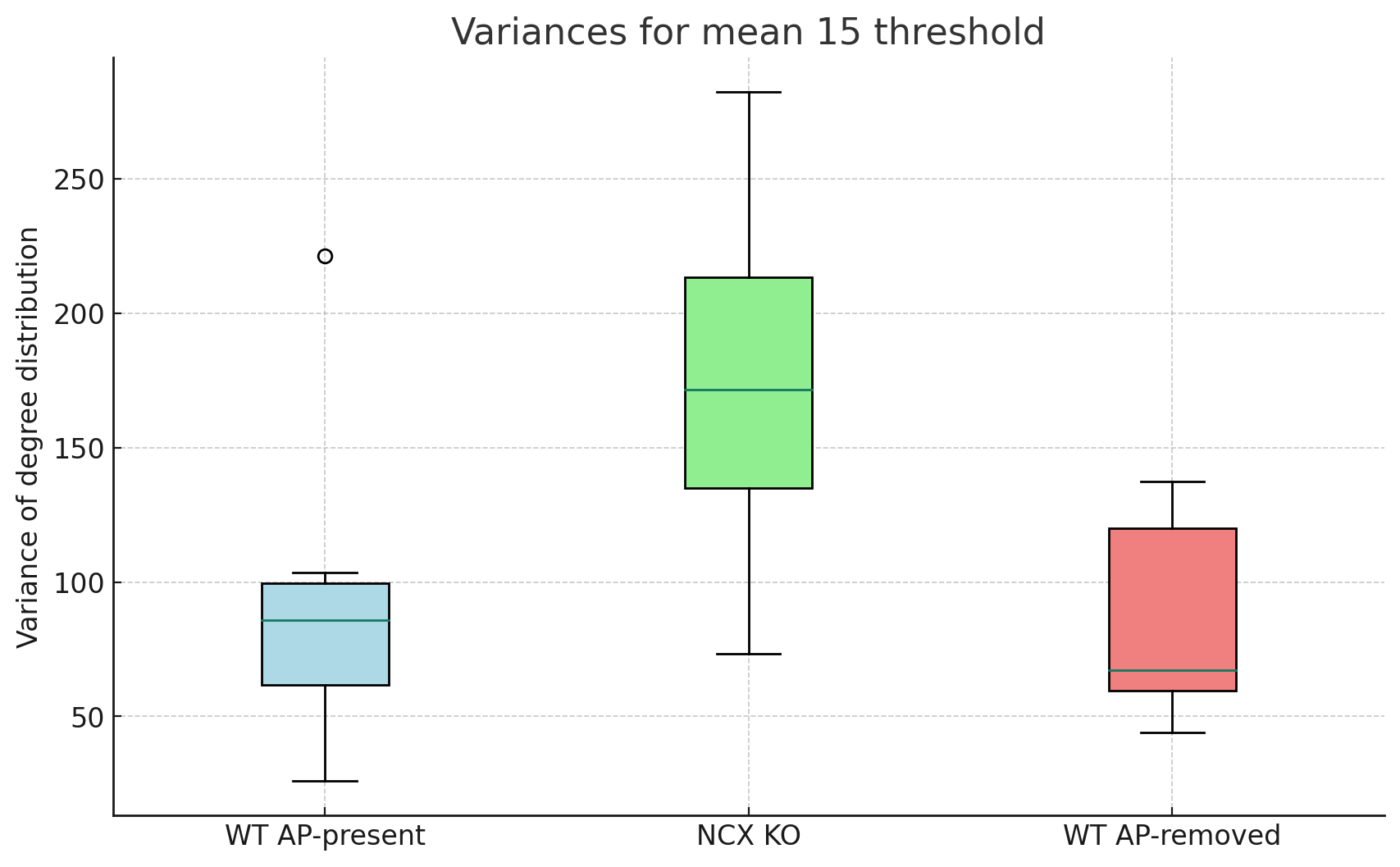}
\caption{The variances of the degree distribution for the three AVN networks }
\label{fig:vars}
\end{figure}

\end{document}